\documentclass[conference]{IEEEtran}
\IEEEoverridecommandlockouts
\usepackage{cite}
\usepackage{amsmath,amssymb,amsfonts,amsthm}
\usepackage{algorithmic}
\usepackage{graphicx}
\usepackage{pgfplots}
\usepgfplotslibrary{groupplots,fillbetween}
\usepackage{textcomp}
\usepackage{xcolor}
\def\BibTeX{{\rm B\kern-.05em{\sc i\kern-.025em b}\kern-.08em
    T\kern-.1667em\lower.7ex\hbox{E}\kern-.125emX}}
\newtheorem{theorem}{Theorem}
\newtheorem{definition}{Definition}
\newtheorem{example}{Example}
\newtheorem{lemma}{Lemma}
\newtheorem{proposition}{Proposition}
\newtheorem{corollary}{Corollary}

\begin{document}

\title{Worst-Case Utility Privacy Mechanism via Pointwise Maximal Leakage\\
\thanks{This work was supported by Vinnova Competence centre Dig-IT lab.}
}

\author{\IEEEauthorblockN{Ci Song, Tobias J. Oechtering}
\IEEEauthorblockA{\textit{Department of Information Science and Engineering} \\
\textit{KTH Royal Institute of Technology, Sweden}}
}

\maketitle

\begin{abstract}
We propose a discrete privacy mechanism exploiting beneficial properties of the novel privacy measure Pointwise Maximal Leakage (PML). Given the utility assignment characterized by every input-output letter pair, we study the mechanism design problem that satisfies PML privacy guarantees and maximizes the worst-case utility. Unlike popular privacy measures like Differential Privacy (DP), PML allows us to set some conditional probabilities in the mechanism to be zero and thereby preventing the occurrence of some low utilities while preserving a strict PML constraint. We show that the utility-safe mechanism, with low computational complexity, is optimal for the worst-case utility problem with an additional constraint on the output support set. We finally demonstrate the effectiveness in several numerical experiments. Due to DP's inability to have zeros in the mechanism, the design of privacy mechanisms that optimize the worst-case utility is underexplored, and this work shows that PML is a privacy measure that is perfectly suited for this purpose.
\end{abstract}


\section{Introduction}
\label{sec:intro}
With the development of data-based technologies in industries, statistical disclosure control has become a critical problem. A privacy mechanism is described by a conditional probability distribution used to disclose information about the data. We study the privacy-utility trade-off in the mechanism design. Besides designs using DP \cite{PrivBook}, e.g., exponential mechanism \cite{ExpM} and geometric mechanism \cite{Geo}, recent information-theoretic mechanism design approaches include \cite{10815813, 9517826, TotalVariationDistance, 9448019}. 

For quantitative privacy risk assessment, the designed mechanism needs to provide a formal privacy guarantee. In this work, we use PML \cite{Saeidian_2023} as the privacy measure. The random variable $X$ is defined on a finite alphabet $\mathcal{X}$ with $|\mathcal{X}|=N$ and PMF $P_X(x)$, and it denotes the data to be protected. $Y$ is defined on the alphabet $\mathcal{Y}$, and denotes the output of the privacy mechanism. Without loss of generality, we assume that $P_X$ has full support. Let $\mathcal{S}_Y$ denote the support of $Y$. It has been shown that PML is operationally meaningful and robust \cite{Saeidian_2023}. The $\epsilon$-PML privacy guarantee for a mechanism is defined as follows.
\begin{definition}[$\epsilon$-PML \cite{Saeidian_2023}]
    Given prior $P_X$, the privacy mechanism $P_{Y|X}$ satisfies $\epsilon$-PML if for all $y \in \mathcal{S}_Y$, we have
    $$
        \ell_{P_{Y|X}\times P_X}(X \rightarrow y) = \log \max_{x \in \mathcal{X}} \tfrac{P_{X|Y=y}(x)}{P_X(x)} \leq \epsilon.
    $$
\end{definition}

Recent risk-management frameworks require us to minimize potential harms to people, especially when certain types of failures can cause greater harm \cite{nist_ai_rmf_2023}. Motivated by accuracy and robustness considerations \cite{eu_ai_act_2024}, the worst-case performance becomes a key design criterion. Different from DP, \footnote{DP does not allow zeros in the mechanism, because DP is designed as the worst-case ratio of conditional probabilities, and zero in the denominator makes the ratio infinite.} PML does not forbid us to set some conditional probabilities $P_{Y|X=x}(y)$ to zero, which enables us to improve the worst-case utility. A mechanism considering the worst-case utility aims to completely avoid undesirable results, providing not only a privacy guarantee but also a utility guarantee. 

Another property of PML is its prior dependence. In \cite[Theorem 4.2]{saeidian2023rethinking}, it has been shown that DP corresponds to the PML guarantee at the worst-case prior distribution that is attained at a degenerated prior where the probability of one letter is one. Thus DP is in general over-conservative, and improvement can be made if some available knowledge about the prior is exploited. In this work, we assume perfect prior knowledge. We recently studied the task of estimating the empirical PML from data \cite{Leo_csf}, which can be directly applied to this setting as well so that this problem is out of our scope. 

Considering infinite $\mathcal{Y}$, the $\epsilon$-PML Laplace mechanism is proposed in \cite{saeidian2023rethinking}. For finite $\mathcal{Y}$ and sub-convex utility functions, the $\epsilon$-PML extremal mechanism is studied in \cite{EMPML}. In this paper, we propose an $\epsilon$-PML mechanism with finite $\mathcal{Y}$ and $|\mathcal{Y}|=M$ considering arbitrary utilities. 

In the following, we will characterize the utility and define the optimization problem in Section \ref{sec:problem}. We will propose the utility-safe mechanism and the general optimal solution in Section \ref{sec:method}. We will compare $\epsilon$-PML utility-safe mechanism with $\epsilon$-PML exponential mechanism and randomized response mechanism in Section \ref{sec:experiment}. Besides, we will also compare the optimality of the utility-safe mechanism and the general optimal solution under different priors in the same section.

\section{Problem Formulation}
\label{sec:problem}
We maximize the worst-case utility while preserving a PML constraint by setting the probabilities of some input-output letter pairs in the mechanism equal to zero to prevent the occurrence of low utilities. To this end, our approach only needs to know the position of the lowest utility values for each input $x\in\mathcal{X}$. This information is preserved in the ascending order of the utilities for each input, and we use this order to represent arbitrary utilities assuming in this work that all utility values are different. Accordingly, let $u: \mathcal{X} \times \mathcal{Y} \rightarrow [M]$ denote the order information of the utility assigned to the output $y$ when the data is $x$, where $[M]$ denotes the set of integers from $1$ to $M$. This means, $u(x_i,y_j)=k$ specifies the pair $(x_i,y_j)$ that results in the $k$-th lowest utility value for input $x_i$.

For simplicity, we use the matrix notation $U = (u_{ij})_{1\leq i \leq N, 1\leq j \leq M}$ for the utility order, where $u_{ij} = u(x_i, y_j)$ and $U_i$ denotes its  $i$-th row. Similarly, we use the matrix notation $P_{Y|X} = (p_{ij})_{1\leq i \leq N, 1\leq j \leq M}$ for the privacy mechanism with $p_{ij} = P_{Y|X=x_i}(y_j)$. For the given prior $P_X$, let $\mathcal{P}_{M|N}^{(\epsilon)}$ denote the set of $\epsilon$-PML privacy mechanisms. As shown in \cite[Lemma 1]{EMPML}, $\mathcal{P}_{M|N}^{(\epsilon)}$ is a polytope in $\mathbb{R}^{N \times M}$ described by linear constraints. Different from the constraints of DP, these constraints allow some parameters to be zero, making it possible to prevent the occurrence of the lowest utilities and thereby enhance the worst-case utility of the mechanism.

For prior probability $P_X$ and utility order matrix $U$, the $\epsilon$-PML worst-case utility mechanism design is formulated as the optimization problem
\begin{gather}
\label{eq:p1}
    \max_{P_{Y|X} \in \mathcal{P}_{M|N}^{(\epsilon)}} h(P_{Y|X}),
\end{gather}
where the objective function $h(P_{Y|X}) = \min_{i,j:p_{ij}>0} u_{ij}$. Let $P_{Y|X}^\perp$ denote the mechanism where the output $Y$ is independent of the input $X$, i.e., $p_{ij}=p_j, \forall i, j$. $P_{Y|X}^\perp$ leaks no information. Thus $P_{Y|X}^\perp \in \mathcal{P}_{M|N}^{(\epsilon)}$ for any $\epsilon \geq 0$ so that problem \eqref{eq:p1} is always feasible.

Problem \eqref{eq:p1} is difficult to compute directly. However, we can utilize properties of PML to efficiently minimize PML when the worst-case utility order is larger than or equal to $h^*$, which is the following problem
\begin{align}
    \label{eq:p2}
    \min_{P_{Y|X} \in \mathcal{P}_{M|N}^{(\infty)}} &\max_{y \in \mathcal{S}_Y} \ell_{P_{Y|X}\times P_X}(X\rightarrow y)\\
    \text{s. t. } &p_{ij}=0, \text{ if }u_{ij}<h^*.\notag
\end{align}
The solution to problem \eqref{eq:p2} will be shown in the next section. 

\begin{lemma}
    \label{le:h}
    Let $h^*(\epsilon) = \max_{P_{Y|X} \in \mathcal{P}_{M|N}^{(\epsilon)}} h(P_{Y|X})$. As $\epsilon$ decreases, $h^*(\epsilon)$ is non-increasing.
    \begin{proof}
        A smaller $\epsilon$ leads to a smaller feasible set for the optimization problem, so the maximum value cannot be larger.
    \end{proof}
\end{lemma}
\noindent
With Lemma \ref{le:h} and the ability to solve problem \eqref{eq:p2}, we can run a binary search over $h^*$. Note that the optimal solution to problem \eqref{eq:p1} is usually not unique. Using problem \eqref{eq:p2} and binary search, we find the mechanism with the lowest PML in the set of optimal solutions to problem \eqref{eq:p1}.

Since $h^*$ is a positive integer less than or equal to $M$, it is possible to enumerate solutions over $h^*$. Thus, we solve problem \eqref{eq:p1} for all $\epsilon \geq 0$ once we have solved problem \eqref{eq:p2} for all $h^*$ and characterized the whole trade-off curve.

\section{Optimal Mechanism Design}
\label{sec:method}
We have discussed how to solve problem \eqref{eq:p1} and further minimize PML without affecting the worst-case utility via solving problem \eqref{eq:p2}. In this Section, we characterize the optimal mechanism of problem \eqref{eq:p2}.

\subsection{Utility-safe Mechanism}
\label{sec:method1}
As a function of conditional probabilities, PML may not be continuous at zero, because the output support $\mathcal{S}_Y$ changes and the new term might decide the worst-case PML level. Because of this, in this subsection we add the following constraint to prevent issues due to the non-continuity and discuss the general case in the following subsection. For a given $h^*$, we require that the output support
\begin{gather}
\label{con:ac}
    \mathcal{S}_Y = \mathcal{Y}^+(h^*)
\end{gather}
where $\mathcal{Y}^+(h^*)=\left \{y_j \in \mathcal{Y}: \exists i \in [N], u_{ij}\geq h^*  \right \}$.

\begin{lemma}[Decomposition Property of PML]
\label{le:twosources}
    Given $P_X$ and $P_{Y|X}$, for all $y \in \mathcal{S}_Y$, we have 
    $$
        \ell_{P_{Y|X}\times P_X}(X\rightarrow y) = \underline{\ell}(\mathcal{S}_{X|Y=y}) + \ell_{P_{Y|X}\times Q_X}(X\rightarrow y),
    $$
    with $\underline{\ell}(\mathcal{S}_{X|Y=y}) = \log \tfrac{1}{\sum_{x \in \mathcal{S}_{X|Y=y}} P_X(x)}$, the rescaled prior $Q_X(x)=\tfrac{P_X(x)}{\sum_{x' \in \mathcal{S}_{X|Y=y}} P_X(x')}$  and the alphabet of $Q_X$ is $\mathcal{S}_{X|Y=y} = \left \{x \in \mathcal{X} : P_{Y|X=x}(y)>0 \right \}$.
    \begin{proof}
    The proof is in \cite[Appendix~C]{EMPML}.
    \end{proof}
\end{lemma}
\noindent
Lemma \ref{le:twosources} shows that PML can be decomposed into two terms. As long as $P_{Y|X=x}(y)$ remains positive, a change of $P_{Y|X=x}(y)$ does not affect the worst-case utility. Only the induced support $\mathcal{S}_{X|Y=y}$ matters. In order to maximize the worst-case utility, we wish to reduce the support $\mathcal{S}_{X|Y=y}$ by reducing the letter that results in the lowest utility. This pruning should be repeatedly done as often as the privacy constraint allows us to do. To enable many pruning steps, we need to minimize $\ell_{P_{Y|X}\times Q_X}(X\rightarrow y)$ so that more privacy budget is saved for $\underline{\ell}(\mathcal{S}_{X|Y=y})$ when reducing the support. 

\begin{lemma}[{Lower bound of PML \cite[Lemma 1]{Saeidian_2023}}]
\label{le:zero}
    With $Q_X(x)$ in Lemma \ref{le:twosources}, $\ell_{P_{Y|X}\times Q_X}(X\rightarrow y) \geq 0$, where the inequality holds with equality if and only if for all $x\in \mathcal{S}_{X|Y=y}$, $ P_{Y|X=x}(y)=P_Y(y)$.
    \begin{proof}
        Here we use $Q_X$ instead of $P_X$. Both are proper prior distributions, so the proof remains the same as in \cite{Saeidian_2023}.
    \end{proof}
\end{lemma}

Let $z_i=\sum_{j=1}^M \boldsymbol{1}(p_{ij}=0)$ denote the number of zeros in $i$-th row of $P_{Y|X}$, where $\boldsymbol{1}(\cdot)$ is the indicator function. For each row $U_i$, the worst-case utility order is upper-bounded by $z_i+1$ with equality if zeros are optimally placed in $P_{Y|X}$. Next, it follows that the worst-case utility order $h(P_{Y|X}) \leq \min_{i\in[N]} z_i+1$. Therefore, our design goal is to achieve $h(P_{Y|X})=z_i+1$ for all $i\in[N]$, which will also allow us to 
meet the conditions in Lemma \ref{le:zero}. 
 
\begin{theorem}[Utility-safe mechanism]
\label{th:rp}
    If $\mathcal{S}_Y=\mathcal{Y}^+(h^*)$, then the optimal mechanism $M^*(h^*)$ of problem \eqref{eq:p2} is
    \begin{gather}
    \label{eq:rpdef}
    [M^*(h^*)]_{ij}= \left \{
    \begin{array}{ll} 
         0 & \text{if }u_{ij}<h^*, \\
         \tfrac{1}{M-h^*+1} & \text{otherwise} ,
    \end{array}
    \right.
    \end{gather}
    where $[M^*(h^*)]_{ij}$ denotes the element in the $i$-th row and $j$-th column of $M^*(h^*)$, i.e., $[M^*(h^*)]_{ij}=\mathbb{P}(Y=y_j|X=x_i)$.
    \end{theorem}
    \begin{proof}
        Let $\mathcal{M}$ be the set of feasible solutions to problem \eqref{eq:p2}. Since $[M^*(h^*)]_{ij}=0$ if and only if $u_{ij}<h^*$, $M^*(h^*) \in \mathcal{M}$. We only need to show that we have $\ell_{M^*(h^*)\times P_X}(X\rightarrow y) \leq \ell_{M\times P_X}(X\rightarrow y)$ for all $y \in \mathcal{Y}^+(h^*)$ and $M \in \mathcal{M}$.
        
        Since $[M^*(h^*)]_{ij}>0$ whenever $u_{ij} \geq h^*$, $\mathcal{S}_{X|Y=y}$ of $M^*(h^*)$ is never smaller than that of $M$. Therefore, with the support set that is always the superset, $M^*(h^*)$ always has the smallest $\underline{\ell}(\mathcal{S}_{X|Y=y})$ in $\mathcal{M}$. Using Lemma \ref{le:zero}, $\ell_{M^*(h^*)\times Q_X}(X\rightarrow y) =0$ for all $y \in \mathcal{Y}^+(h^*)$. Using Lemma \ref{le:twosources}, for all $y \in \mathcal{Y}^+(h^*)$,
        \begin{gather}
        \label{eq:smallest}
            \ell_{M^*(h^*)\times P_X}(X\rightarrow y) \leq \ell_{M\times P_X}(X\rightarrow y),
        \end{gather}
        so $M^*(h^*)$ is an optimal solution to problem \eqref{eq:p2}. 
    \end{proof}

\begin{corollary}
\label{co:value}
    The optimal value of $\epsilon$ achieved by Theorem \ref{th:rp} is $- \log \min_{y \in \mathcal{Y}^+(h^*)} \sum_{x \in \mathcal{X}: u(x, y) \geq h^*} P_X(x)$.
\end{corollary}
According to Theorem \ref{th:rp} and Corollary \ref{co:value}, the proposed mechanism only requires comparing the utility order against the threshold and the summation of prior probabilities, so it has very low computational complexity.
 
\begin{corollary}
\label{co:rp}
    Using $M^*(h^*)$ enforces that  we have 
$$
\tfrac{P_X(x)}{P_X(x')} = \tfrac{P_{X|Y=y}(x)}{P_{X|Y=y}(x')}, \quad \forall x, x' \in \mathcal{S}_{X|Y=y}.
$$
\end{corollary}    
    \begin{proof}
   Due to \eqref{eq:rpdef}, $P_{Y|X=x}(y)$ is  the same for all $x \in \mathcal{S}_{X|Y=y}$ so that
        $
            \tfrac{P_{X|Y=y}(x)}{P_{X|Y=y}(x')} = \tfrac{P_{Y|X=x}(y) P_X(x) / P_Y(y)}{P_{Y|X=x'}(y) P_X(x') / P_Y(y)}=\tfrac{P_X(x)}{P_X(x')}.
        $
    \end{proof}
\noindent
Given the utility-safe mechanism and the output $y$, the adversary knows that some $x \in \mathcal{X}$ are impossible to be realized. This is the necessary cost for low utility prevention. Moreover, Corollary \ref{co:rp} shows that the ratio between the probabilities of possible $x$ remains the same a priori and a posteriori. Thus, the utility-safe mechanism reveals $\mathcal{S}_{X|Y=y}$ and nothing more. 
If $\mathcal{S}_Y=\mathcal{Y}^+(h^\ast)$, it achieves the minimal privacy leakage. However, if $|S_Y|$ is changed, lower leakage can be achieved in some cases, which will be shown next.  

\subsection{General Optimal Solution}
\label{sec:method2}
First, we want to show that the utility-safe mechanism is not always the optimal mechanism of problem \eqref{eq:p2}. The following example illustrates the case where a change in $\mathcal{S}_Y$ may allow a lower PML.
\begin{example}
\label{ex:im}
    Let
    \begin{gather}
        U = 
        \begin{bmatrix}
        3 & 2 & 1 \\
        1 & 3 & 2 \\
        2 & 1 & 3
    \end{bmatrix}.
    \end{gather}
    We assume that $x_1 = \arg \max_{x \in \mathcal{X}} P_X(x)$. When $P_X(x_2)+ P_X(x_3) \leq \tfrac{1}{2}$ and $\log 2 \leq \epsilon < $ $ - \log \min_{x \in \mathcal{X}} P_X(x)$, it can be shown that an optimal solution is
    \begin{gather}
        \label{eq:newsol}
        P_{Y|X} = 
        \begin{bmatrix}
        \tfrac{1-2P_X(x_3)}{2P_X(x_1)} & \tfrac{1-2P_X(x_2)}{2P_X(x_1)} & 0 \\
        0 & 1 & 0 \\
        1 & 0 & 0
        \end{bmatrix}.
    \end{gather}

    When $P_X(x_2)+ P_X(x_3) \geq \tfrac{1}{2}$ and $- \log (P_X(x_2)+ P_X(x_3)) \leq \epsilon < - \log \min_{x \in \mathcal{X}} P_X(x)$, it can be shown that an optimal mechanism is $M^*(2)$.
\end{example}
\noindent
Thus, a smaller $\mathcal{S}_Y$ may allow to achieve lower PML, especially when $\underline{\ell}(\mathcal{S}_{X|Y=y_j})$ is much higher than other $y \in \mathcal{Y} \backslash \left \{ y_j \right \}$. In Example \ref{ex:im}, $\underline{\ell}(\mathcal{S}_{X|Y=y_3})$ is larger than $\underline{\ell}(\mathcal{S}_{X|Y=y_1})$ and $\underline{\ell}(\mathcal{S}_{X|Y=y_2})$ when using $M^*(2)$. Setting $P_Y(y_3)=0$ increases $\ell_{P_{Y|X}\times Q_X}(X\rightarrow y_1)$ and $\ell_{P_{Y|X}\times Q_X}(X\rightarrow y_2)$. However, since  $\epsilon$-PML  is the worst case over $\mathcal{S}_Y$, the privacy guarantee is not decided by $y_3$ any more, so that the $\epsilon$-PML privacy guarantee improves.

Next, we briefly show that we can also solve problem \eqref{eq:p2} via linear programming solvers if the problem complexity is sufficiently small. Although problem \eqref{eq:p2} is not a linear programming problem, we can find the optimal solution by checking the feasibility of the constraints
\begin{gather}
    p_{ij}=0, \text{ if }u_{ij}<h^*\notag\\
    p_{ij} - (\sum_{k=1}^N P_X(x_k)p_{kj})\exp(\epsilon) \leq 0.
    \label{eq:p3}
\end{gather}
This checks whether we can design an $\epsilon$-PML mechanism and keep the worst-case utility no less than $h^*$ at the same time, i.e., if $\left \{ P_{Y|X} \in \mathcal{P}_{M|N}^{(\epsilon)}: p_{ij}=0, \text{ if }u_{ij}<h^* \right \}$ is non-empty for utility $h^*$ and privacy level $\epsilon$. If so, we can attempt a lower $\epsilon$. If not, we need to increase $\epsilon$. We can use the bisection method to find the optimal $\epsilon$.

Using the conclusion from Example \ref{ex:im}, the number of variables in problem \eqref{eq:p3} can be further reduced. For all $y \in \mathcal{Y}$, we enforce that $P_Y(y)=0$ if $- \log \sum_{x \in \mathcal{X}: u(x, y) \geq h^*} P_X(x)$ exceeds $\epsilon$. This is because the PML of an output $y$ is lower bounded by $- \log \sum_{x \in \mathcal{X}: u(x, y) \geq h^*} P_X(x)$, and we cannot make it lower unless $y$ is not in $\mathcal{S}_Y$ anymore. 

Problem \eqref{eq:p3} has $O(NM)$ variables and constraints, so the computational complexity grows polynomially fast and quickly becomes computationally exhaustive. In contrast, the computational complexity of the proposed utility-safe mechanism scales linearly with the predefined dimensions.

\section{Numerical Experiment}
\label{sec:experiment}
In the first subsection, we compare the worst-case performance of the proposed utility-safe mechanism against some classic DP-based mechanisms under a PML guarantee. The numerical study on the optimality of the utility-safe mechanism is done in the next subsection.

\subsection{Worst-case Utility in Counting Query}
The exponential mechanism \cite{PrivBook} and randomized response mechanism \cite{RRM} are typical discrete output privacy mechanisms, mainly designed for DP or Local Differential Privacy (LDP) \cite{LDP1} \cite{LDP}. The design of such mechanisms for PML can be achieved via the relationship between LDP and PML, which is shown in \cite[Proposition~6]{Saeidian_2023}.
\begin{proposition}[{\cite[Proposition~6]{Saeidian_2023}}]
\label{pr:ldp}
    Let $p_{\text{min}} = \min_{x \in \mathcal{X}} \allowbreak P_X(x) \in (0, 1)$. When $\epsilon \geq - \log p_{\text{min}}$, any mechanism $P_{Y|X}$ satisfies $\epsilon$-PML. When $\epsilon < - \log p_{\text{min}}$, the mechanism $P_{Y|X}$ satisfies $\epsilon$-PML if it satisfies $\overline{\epsilon}$-LDP with
    $$
        \overline{\epsilon} \leq - \log \tfrac{\exp{(-\epsilon)-p_{\text{min}}}}{1-p_{\text{min}}}.
    $$
\end{proposition}

In our numerical experiment, we solve a counting query problem of $6$ data records, and let $X$ denote the number of records that satisfy the designated property. We assume that $P_X(x_i)=\tfrac{1}{7}, \forall i =1, \dots, 7$, where $x_i=i-1$, and the output alphabet is the same as the input, which is $\mathcal{Y}=\mathcal{X}=\left \{ 0, \dots, 6 \right \}$. The utility corresponds to a quadratic loss with an additional penalty of one if the disclosed value is lower than the true value, i.e.,
\begin{align}
\label{eq:uval}
U' & = (u'_{ij})_{1\leq i \leq 7, 1\leq j \leq 7}\notag\\
&=
    \begin{bmatrix}
        0 & -1 & -4 & -9 & -16 & -25 & -36\\
        -2 & 0 & -1 & -4 & -9 & -16 & -25\\
        -5 & -2 & 0 & -1 & -4 & -9 & -16\\
        -9 & -5 & -2 & 0 & -1 & -4 & -9\\
        -17 & -9 & -5 & -2 & 0 & -1 & -4\\
        -26 & -17 & -9 & -5 & -2 & 0 & -1\\
        -37 & -26 & -17 & -9 & -5 & -2 & 0
    \end{bmatrix}.
\end{align}
The corresponding order information is preserved in the utility order matrix
    \begin{gather}
    \label{eq:umtx}
        U = 
        \begin{bmatrix}
        7 & 6 & 5 & 4 & 3 & 2 & 1\\
        5 & 7 & 6 & 4 & 3 & 2 & 1\\
        3 & 5 & 7 & 6 & 4 & 2 & 1\\
        1 & 3 & 5 & 7 & 6 & 4 & 2\\
        1 & 2 & 3 & 5 & 7 & 6 & 4\\
        1 & 2 & 3 & 4 & 5 & 7 & 6\\
        1 & 2 & 3 & 4 & 5 & 6 & 7
    \end{bmatrix}.
    \end{gather}
\noindent

The common DP solution to the counting query problem is the Laplace mechanism, which does not have a lower bound on the utility. We study discrete output mechanisms, so in this case, we use the $\epsilon$-PML exponential mechanism 
$$
P^{E}_{Y|X=x_i}(y_j) \allowbreak = \tfrac{\exp{(\tfrac{u'_{ij}\overline{\epsilon}}{74})}}{\sum_j \exp{(\tfrac{u'_{ij}\overline{\epsilon}}{74})}}.
$$ 
We use the $\epsilon$-PML randomized response mechanism to generate $x'$ from $x$, and $y = \arg \max_{y \in \mathcal{Y}} u(x', y)$. This is
\begin{gather*}
     P^{R}_{Y|X=x_i}(y_j) = \left \{
    \begin{array}{ll} 
        \tfrac{e^{\overline{\epsilon}}}{6+e^{\overline{\epsilon}}} & \text{if }i=j, \\
        \tfrac{1}{6+e^{\overline{\epsilon}}} & \text{otherwise }.
    \end{array}
    \right.
\end{gather*}
Depending on the value of $\epsilon$, the utility-safe mechanism is
\begin{gather*}
    P^{S}_{Y|X=x_i}(y_j) = \left \{
    \begin{array}{ll} 
        M^*(1) & \text{if }0\leq \epsilon < \log \tfrac{7}{3}, \\
        M^*(3) & \text{if }\log \tfrac{7}{3} \leq \epsilon < \log \frac{7}{2}, \\
        M^*(5) & \text{if }\log \frac{7}{2} \leq \epsilon < \log 7, \\
        M^*(7) & \text{if }\epsilon \geq \log 7.
    \end{array}
    \right.
\end{gather*}

We increase $\epsilon$ from $0.5$ to $2$ with increments of $0.05$. For each $\epsilon$, we generate $X$ according to the prior distribution $P_X$, choose $Y$ according to the mechanism, and repeat the experiment 1000 times. Fig. \ref{fig:ube} depicts the minimum utility value. When the privacy level $\epsilon$ is very low, we are in the high privacy regime where PML also cannot place any zero in the $P_{Y|X}$ so that we have no worst-case utility improvement. However, there are sharp increases in the worst-case utility for the utility-safe mechanism at $\epsilon=0.85$, $1.30$ and $1.95$, while the minimum values for the DP-based mechanisms (exponential (E) \& randomized response (R)) do not increase at all, which is as expected. Moreover, the increase in the worst-case utility value depends on the utility matrix. For example, $x=0, y=6$ leads to a very low utility, and preventing it brings a giant leap in the worst-case utility value. 
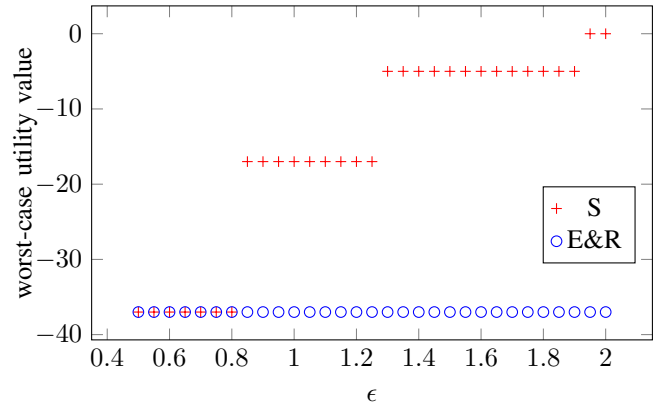
\begin{figure}[h]
  \centering
\begin{tikzpicture}
\begin{groupplot}[
    group style={group size=1 by 1},
    width=9cm, height=6cm,
    xlabel={$\epsilon$},
    ylabel={worst-case utility value},
    ylabel style={yshift=-10pt}
]

\nextgroupplot[legend pos=south east,
legend style={yshift=25pt} 
]]
\addplot[name path=min1, red, thin, only marks, mark=+] coordinates {
    (0.5,-37) (0.55,-37) (0.6,-37) (0.65,-37) (0.7,-37) (0.75,-37) (0.8,-37) (0.85,-17) (0.9,-17) (0.95,-17) (1.0,-17) (1.05, -17) (1.1,-17) (1.15, -17) (1.2,-17) (1.25, -17) (1.3,-5) (1.35, -5) (1.4, -5) (1.45, -5) (1.5, -5) (1.55, -5) (1.6, -5) (1.65, -5) (1.7, -5) (1.75, -5) (1.8, -5) (1.85, -5) (1.9, -5) (1.95, 0) (2.0, 0)
};
\addlegendentry{S}

\addplot[name path=min2, blue, thin, only marks, mark=o] coordinates {
    (0.5,-37) (0.55,-37) (0.6,-37) (0.65,-37) (0.7,-37) (0.75,-37) (0.8,-37) (0.85,-37) (0.9,-37) (0.95,-37) (1.0,-37) (1.05, -37) (1.1,-37) (1.15, -37) (1.2,-37) (1.25, -37) (1.3,-37) (1.35, -37) (1.4, -37) (1.45, -37) (1.5, -37) (1.55, -37) (1.6, -37) (1.65, -37) (1.7, -37) (1.75, -37) (1.8, -37) (1.85, -37) (1.9, -37) (1.95, -37) (2.0, -37)
};
\addlegendentry{E\&R}

\end{groupplot}
\end{tikzpicture}
  \caption{Sample minimum of utility given $x$ when using $\epsilon$-PML utility-safe (S), exponential (E), and randomized response (R) mechanisms. Utility-safe mechanism offers higher worst-case utility when the privacy level $\epsilon$ is not too low (here $\epsilon>0.8$). Other mechanisms do not improve their worst-case utility with increasing $\epsilon$.}
  \label{fig:ube}
\end{figure}

\subsection{Optimality}
The utility-safe mechanism is optimal if $\mathcal{S}_Y$ satisfies the constraint in \eqref{con:ac}. As is shown in subsection \ref{sec:method2}, in some rare situations this prior assumption results in a suboptimal mechanism, which is further discussed in this subsection.
    \begin{gather}
    \label{eq:umtx2}
        U = 
        \begin{bmatrix}
        4 & 3 & 2 & 1\\
        1 & 4 & 3 & 2\\
        2 & 1 & 4 & 3\\
        3 & 2 & 1 & 4
    \end{bmatrix}.
    \end{gather}
\noindent
We use the utility order matrix $U$ in \eqref{eq:umtx2} so that the priors $P_X(x_i)$ are interchangeable. The prior distribution follows two patterns. The `one-low-three-high' pattern in Fig. \ref{fig:onelow} means that one prior probability is $p_{\text{min}}$ and the other three are equal. The `three-low-one-high' pattern in Fig. \ref{fig:threelow} means that three prior probabilities are $p_{\text{min}}$. We only plot the optimal $\epsilon$ using the utility-safe mechanism in Subsection \ref{sec:method1} and the general optimal solution in Subsection \ref{sec:method2} when $h^*=2$ and $h^*=3$, because the cases of $h^*=1$ and $h^*=4$ are naive. In both figures, we increase $p_{\text{min}}$ from $0.02$ to $0.2$ with increments of $0.02$. The red lines are the optimal PML given the additional constraint on $\mathcal{S}_Y$, which is calculated from Corollary \ref{co:value}, while the blue lines are optimal without the additional constraint according to linear programming solvers. The proposed utility-safe mechanism remains optimal even without the additional constraint \eqref{con:ac} if the two lines are the same. 

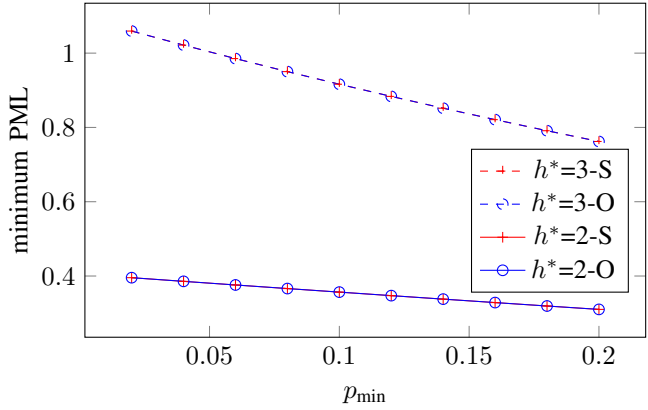
\begin{figure}[h]
  \centering
\begin{tikzpicture}
\begin{groupplot}[
    group style={group size=1 by 1},
    width=9cm, height=6cm,
    xticklabel style={/pgf/number format/fixed},
    xlabel={$p_{\text{min}}$},
    ylabel={minimum PML},
    ylabel style={yshift=-10pt}
]

\nextgroupplot[legend pos=south east,
legend style={yshift=13pt} 
]]
\addplot[name path=h3s, red, thin, dashed, mark=+] coordinates {
    (0.02, 1.05939158) (0.04, 1.02165125) (0.06, 0.9852836) (0.08, 0.95019228) (0.10, 0.91629073) (0.12, 0.88350091) (0.14, 0.85175221) (0.16, 0.82098055) (0.18, 0.79112759) (0.20, 0.76214005)
};
\addlegendentry{$h^*$=3-S}

\addplot[name path=h3o, blue, thin, dashed, mark=o] coordinates {
(0.02, 1.05939865) (0.04, 1.02165222) (0.06, 0.98529816) (0.08, 0.95020294) (0.10, 0.91630936) (0.12, 0.88350296) (0.14, 0.85176468) (0.16, 0.82099915) (0.18, 0.79113007) (0.20, 0.76215744)
};
\addlegendentry{$h^*$=3-O}

\addplot[name path=h2s, red, thin, mark=+] coordinates {
    (0.02, 0.39551478) (0.04, 0.38566248) (0.06, 0.37590631) (0.08, 0.36624439) (0.10, 0.35667494) (0.12, 0.34719620) (0.14, 0.33780646) (0.16, 0.32850407) (0.18, 0.31928741) (0.20, 0.31015493)
};
\addlegendentry{$h^*$=2-S}

\addplot[name path=h2o, blue, thin, mark=o] coordinates {
(0.02, 0.39552689) (0.04, 0.38566589) (0.06, 0.37591934) (0.08, 0.36624908) (0.10, 0.35669327) (0.12, 0.34721375) (0.14, 0.33781052) (0.16, 0.32852173) (0.18, 0.31929016) (0.20, 0.31017303)
};
\addlegendentry{$h^*$=2-O}

\end{groupplot}
\end{tikzpicture}
  \caption{Lowest leakage under worst-case utility $h^*$ and `one-low-three-high' prior with the utility-safe mechanism (S) and optimal mechanism (O). Utility-safe mechanism achieves optimal performance.}
  \label{fig:onelow}
\end{figure}

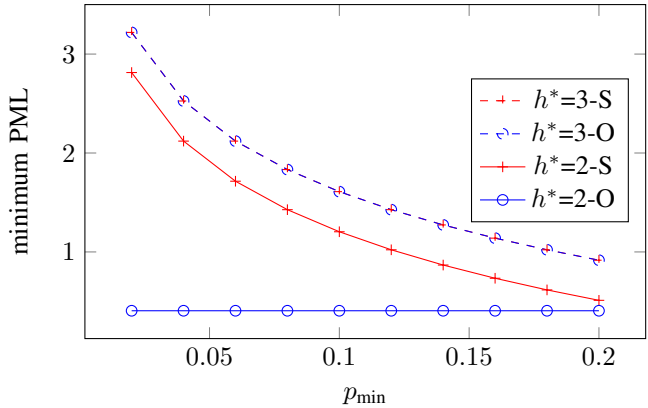
\begin{figure}[h]
  \centering
\begin{tikzpicture}
\begin{groupplot}[
    group style={group size=1 by 1},
    width=9cm, height=6cm,
    xticklabel style={/pgf/number format/fixed},
    xlabel={$p_{\text{min}}$},
    ylabel={minimum PML},
    ylabel style={yshift=-10pt}
]

\nextgroupplot[legend pos=south east,
legend style={yshift=40pt} 
]]
\addplot[name path=h3s, red, thin, dashed, mark=+] coordinates {
(0.02, 3.21887582) (0.04, 2.52572864) (0.06, 2.12026354) (0.08, 1.83258146) (0.10, 1.60943791) (0.12, 1.42711636) (0.14, 1.27296568) (0.16, 1.13943428) (0.18, 1.02165125) (0.20, 0.91629073)
};
\addlegendentry{$h^*$=3-S}

\addplot[name path=h3o, blue, thin, dashed, mark=o] coordinates {
(0.02, 3.2188797) (0.04, 2.52573013) (0.06, 2.12026596) (0.08, 1.83259964) (0.10, 1.60943985) (0.12, 1.42711639) (0.14, 1.27298355) (0.16, 1.13945007) (0.18, 1.02165222) (0.20, 0.91630936)
};
\addlegendentry{$h^*$=3-O}

\addplot[name path=h2s, red, thin, mark=+] coordinates {
(0.02, 2.81341072) (0.04, 2.12026354) (0.06, 1.71479843) (0.08, 1.42711636) (0.10, 1.2039728) (0.12, 1.02165125) (0.14, 0.86750057) (0.16, 0.73396918) (0.18, 0.61618614) (0.20, 0.51082562)
};
\addlegendentry{$h^*$=2-S}

\addplot[name path=h2o, blue, thin, mark=o] coordinates {
(0.02, 0.40548325) (0.04, 0.40548325) (0.06, 0.40548325) (0.08, 0.40548325) (0.10, 0.40548325) (0.12, 0.40548325) (0.14, 0.40548325) (0.16, 0.40548325) (0.18, 0.40548325) (0.20, 0.40548325)
};
\addlegendentry{$h^*$=2-O}

\end{groupplot}
\end{tikzpicture}
  \caption{Lowest leakage under worst-case utility $h^*$ and `three-low-one-high' prior with the utility-safe mechanism (S) and optimal mechanism (O). Suboptimality can be observed in low worst-case utility case (solid lines), while our utility-safe mechanism is optimal for higher utilities (dashed lines).}
  \label{fig:threelow}
\end{figure}

The optimality of utility-safe mechanisms depends on many factors, including $h^*$, the pattern and value of the prior distribution, and no single factor alone determines the sub-optimality of the utility-safe mechanism. According to Fig. \ref{fig:onelow}, in the `one-low-three-high' pattern the utility-safe mechanism is always optimal. Fig. \ref{fig:threelow} exhibits that when $h^*=3$ the utility-safe mechanism is optimal, but it is not optimal when $h^*=2$. The difference between the two solutions increases as all the prior probabilities concentrate more on one letter.

\section{Conclusion}
\label{sec:conclusion}

In contrast to DP, an $\epsilon$-PML privacy mechanism allows zeros in $P_{Y|X}$ if $\epsilon$ is large enough. This property can be exploited for an improved worst-case utility performance as demonstrated in our work. Using the decomposition property of PML, we identify the maximal number of zeros that can be placed in $P_{Y|X}$ given a PML constraint and propose the utility-safe mechanism. Numerical experiments illustrate that the mechanism leads to improved worst-case utility performance. Our proposed utility-safe mechanism often achieves the optimal performance that can be computed by using a linear programming solver only if the dimensions are sufficiently small. Experiments show that if the prior becomes concentrated, then improved designs might be possible.

\bibliographystyle{IEEEtran}
\bibliography{refs}

\end{document}